\documentclass{llncs}
\usepackage{llncsdoc}
\usepackage{algorithm}
\usepackage[noend]{algpseudocode}
\usepackage{amsmath} 
\usepackage{graphicx}
\usepackage{amssymb}  

\usepackage{thmtools,thm-restate}

\author{Pratik Ghosal and Katarzyna Paluch}

\date{}

\institute{University of Wroclaw}
\title{Manipulation  Strategies for the Rank-Maximal Matching Problem\protect\footnote{Partly supported by Polish National Science Center grant UMO-2013/11/B/ST6/01748}}
\begin{document}
\maketitle
\thispagestyle{empty}

\begin{abstract}
We consider manipulation strategies for the rank-maximal matching problem. Let $G = (A \cup P, \mathcal{E})$ be a bipartite graph  such that $A$ denotes a set of applicants and $P$ a set of  posts. Each applicant $a \in A$ has a preference list over the set of his neighbours in $G$, possibly involving ties.  
A matching $M$ is any subset of edges from $\mathcal{E}$ such that no two edges of $M$ share an endpoint.  A {\em rank-maximal} matching is one in which the maximum number of applicants is matched to their rank one posts, subject to this condition, the maximum number of applicants is matched to their rank two posts and so on. 
A central authority matches  applicants to  posts in $G$ using one of  rank-maximal matchings.   Let $a_1$ be the sole manipulative applicant, who knows the preference lists of all the other applicants and wants to falsify his preference list, so that, he has a chance of getting better posts than if he were truthful, i.e., than if he gave a true preference list. 

We give three manipulation strategies for $a_1$ in this paper. In the first problem `best nonfirst', the manipulative applicant $a_1$ wants to ensure that he is never matched to any post worse than the most preferred post among those of rank greater than one  and obtainable, when he is truthful. In the second strategy `min max' the manipulator wants to construct  a preference list for $a_1$ such that the worst post he can become matched to by the central authority is best possible or in other words, $a_1$ wants to minimize the maximal rank of a post he can become matched to. %`min max' is never worse strategy for $a_1$ than `best nonfirst' and sometimes `min max' is better for $a_1$ than `best nonfirst'. 
To be able to carry out strategy `best nonfirst', $a_1$ only needs to know  the most preferred post of each applicant,  whereas putting into effect `min max' requires the knowledge of whole preference lists of all applicants. The last manipulation strategy `improve best' guarantees that $a_1$ is matched  to his most preferred post at least in some rank-maximal matchings.

\end{abstract}
\section{Introduction} \label{introduction}
We consider manipulation strategies for the rank-maximal matching problem. In the  rank-maximal matching problem,  we are given a bipartite graph $G = (A \cup P, \mathcal{E})$ where $A$ denotes a set of applicants and $P$ a set of  posts. Each applicant $a \in A$ has a preference list over the set of his neighbours in $G$, possibly involving ties. Preference lists are represented by ranks on the edges - an edge $(a,p)$ has rank $i$, denoted as $rank(a,p)=i$, if post $p$ belongs to one of $a$'s $i$-th choices. An applicant $a$ prefers a post $p$ to a post $p'$ if $rank(a,p)<rank(a,p')$. In this case, we say that $(a,p)$ has higher rank than $(a,p')$. If $a$ is indifferent between $p$ and $p'$, then $rank(a,p)=rank(a,p')$. Posts most preferred by an applicant $a$
have rank one in his preference list. A {\em matching} $M$ is any subset of edges $\mathcal{E}$ such that no two edges of $M$ share an endpoint.  A matching is called a {\em rank-maximal} matching if it matches the maximum number of applicants to their rank one posts and subject to this condition, the maximum number of applicants  to their rank two posts, and so on. A rank-maximal matching can be computed in $O(\min(c \sqrt{n},n) m)$ time, where $n$ is the number of applicants, $m$ the number of edges and $c$ the maximum rank of an edge in an optimal solution \cite{IrvingKMMP06}. 

A central authority matches applicants to  posts by using the rank-maximal matching algorithm.  Since there may be more than one rank- maximal matching of $G$, we assume that  the central authority may choose any one of them arbitrarily. Let $a_1$ be a manipulative applicant, who knows the preference lists of all the other applicants and wants to falsify his preference list, so that, he has a chance of getting better posts than if he were truthful, i.e., than if he gave a true preference list. We can always assume that $a_1$ does not get his most preferred post in every rank-maximal matching when he is truthful, otherwise, $a_1$ does not have any incentive to cheat. Also, we can notice that it is usually advantageous for $a_1$ to truncate his preference list. Let $H_p$ denote the graph, in which $a_1$'s
preference list consists of only one post $p$. Then as long as no rank-maximal matching of $H_p$ leaves $a_1$ unmatched, he is guaranteed  to
always get the post $p$. To cover the worst case situation for $a_1$, our strategies require $a_1$ to provide a full preference list that includes every post from $P$. 
%To prevent such procedures, we require that $a_1$ and generally each applicant is required to provide a full preference list that includes every post from $P$. 
Also, $a_1$ could make  the posts, he does not want to be matched to, appear very  far in his preference list.  Thus, we  assume that $a_1$ does not have any gap in his  preference list, i.e., it cannot happen that in $a_1$'s preference list there are a rank $i$ and rank $(i+2)$ posts but none of rank $(i+1)$.

\textbf{Our Contribution:} Our  contribution consists in  developing  manipulation strategies for the rank-maximal matching problem. Given a graph instance with the true preference list of every applicant, we introduce three manipulation strategies for $a_1$. We consider the case where $a_1$ is the sole manipulator in $G$.

 Our first manipulation strategy named `best nonfirst' is described in Section \ref{true}. The strategy may not provide an optimal improvement for $a_1$, but it is  simple and fast. This strategy guarantees that $a_1$ is never matched to any post worse than the second best  post he can be matched to in a rank-maximal matching, when he is truthful. In other words, if $a_1$ is matched to a post $p$ when he is truthful and $p$ is not his most preferred post, then the strategy `best nonfirst' ensures that he is never matched to any post ranked worse than $p$ in any rank-maximal matching. The advantage of this strategy is that $a_1$ does not need to know  full preference lists of the other applicants. He only needs to know  the most preferred post of each applicant to be able to successfully execute the strategy.

Next, in  Section \ref{minmax} we propose the strategy `min max'. The strategy minimizes the maximal rank of a post $a_1$ can become matched to. Thus it optimally improves the worst post of $a_1$ that is obtainable from the central authority. What is more, the strategy has the property that by using it, $a_1$ always gets matched to $p_1$, which is the best among  worst posts he can be matched to. Moreover, we prove that there does not exist a strategy that simultaneously guarantees that $a_1$ never gets a post worse than $p_1$ and sometimes gets a post better than $p_1$.

Last but not least, we have studied the manipulation strategy `improve best' in Section \ref{improvebest}. The previous two manipulation strategies improve the worst post $a_1$ can be matched to in a rank-maximal matching. Hence, these strategies may not match $a_1$ to his most preferred post in any rank maximal matching. In this manipulation strategy, $a_1$ has a different goal - he wants to be matched to his most preferred post in some rank-maximal matchings. Note that it is not possible for him to ensure that he always gets his most preferred post. 

%By a  strategy termed `best nonfirst'(section \ref{true})	 we denote any graph $H \in \cal{H}$ such that every rank-maximal matching of $H$ matches $a_1$ to a post  most preferred among those of rank greater than one  and obtainable when he is truthful. Such posts are the most preferred non-$f$-posts of $a_1$ in $G$. (An  $f$-post is defined in Section \ref{fpost}).  And by a strategy termed `min max' we denote any graph $H \in \cal{H}$ such that every rank-maximal matching of $H$ matches $a_1$ to a post having rank equal to $\min\{\max(Z(H)): H \in \cal{H}\}$. Our algorithm of `min max' strategy matches $a_1$ to a particular post $p$ in rank maximal matching of $H$. We also show that the output of any optimal `min max' strategy also matches $a_1$ to that same post $p$ in every rank maximal matching. 

{\bf Previous and related work.}
The rank-maximal matching problem belongs to the class of matching problems with preferences. In the problems with one-sided preferences, the considered graph
is bipartite and each vertex of only  one set of the bipartition expresses  preferences over the set of its neighbours. Apart from rank-maximal matchings, other types of 
matchings from this class include  pareto-optimal \cite{abdulkadirouglu1998random} \cite{roth1977weak} \cite{AzizBH13}, popular \cite{AbrahamIKM07} and fair\cite{huang2016fair} matchings among others. In the problems with two-sided preferences, the underlying graph is also bipartite but  vertices  from both  sides of the bipartition express preferences  over their  neighbours.  The most famous example of a matching problem with two-sided preferences is that of a  stable matching known also as the stable marriage problem. Since the seminal paper by Gale and Shapley \cite{1962GaleShapley}, it has been studied very intensively, among others in \cite{1989GusfieldIrving},\cite{1998Irving},\cite{1984Roth}. In  the non-bipartite matching  problems with preferences each vertex from the graph ranks all of its neighbours.  The stable roommate problem \cite{irving1985efficient} is a counterpart of the stable marriage problem in the non-bipartite setting.

The rank-maximal matching problem was first introduced by Irving\cite{Irvgreedy}. A rank-maximal matching can be found via a relatively straightforward reduction to the maximum weight matching problem.   The already mentioned \cite{IrvingKMMP06} gives
a combinatorial algorithm that runs in $O(\min(n, c \sqrt n)m)$ time. The capacitated and weighted versions were considered, respectively,
in \cite{paluch2013capacitated} and \cite{kavitha2006efficient}. A switching graph characterization of the set of all rank-maximal matchings is described in \cite{ghosal2014rank}. Finally, the dynamic version of the rank-maximal matching problem was considered in \cite{NimbhorkarV17} and \cite{GhosalKP17}. 

  A matching problem with preferences is called strategy-proof if it is in the best interest of each applicant to provide their true preference list.  An example of a strategy-proof mechanism among  matching problems with one-sided preferences is that of a pareto optimal matching. The strategyproofness of a pareto optimal matching has applications  in house allocation  \cite{hylland1979efficient} \cite{AbrahamCMM05}  \cite{shapley1974cores} \cite{KrystaMRZ14} and kidney exchange  \cite{RothSU05} \cite{AshlagiFKP15}. Regarding the stable matching problem, if a stable matching algorithm produces a  men-optimal stable matching, then it is not possible for men to gain any advantage  by changing or contracting their preference lists and then the best strategy for them is to keep their true preference lists \cite{dubins1981machiavelli}\cite{roth1982economics}.

In the context of matching with  preferences  cheating strategies were mainly studied for the stable matching problem.  Gale and Sotomayor \cite{gale1985ms} showed that women can shorten their preference lists to force an algorithm, that computes the men-optimal stable matching, to produce the  women-optimal stable matching. Teo et al. \cite{teo2001gale} considered a cheating strategy, where  women are required to give a full preference list and one of the women is a manipulator.  Huang \cite{huang2006cheating} explored the versions, in which, men can make coalitions. Manipulation strategy in the  stable roommate problem was also considered by Huang\cite{huang2007cheating}. For a matching problem with one-sided preferences, Nasre\cite{nasre2013popular} studied manipulation strategies for the popular matching problem.

\section{Background} \label{background}
  
A matching $M$ is said to be {\em maximum (in a graph $G$)} if, among all matchings of $G$, it has the maximum number of edges.
A path $P$ is said to be {\em alternating with respect to matching $M$} or {\em $M$-alternating} if its edges  belong alternately to $M$ and $\mathcal{E} \setminus M$.  A vertex $v$ is {\em unmatched} or {\em free} in $M$ if it is not incident to any edge of $M$. An $M$-alternating path $P$ such that both its endpoints are unmatched in $M$, is said to be {\em $M$-augmenting} (or augmenting with respect to $M$). It was proved by Berge \cite{berge1957two} that a matching $M$ is maximum if and only if there exists no $M$-augmenting path.

We  state the following well-known properties of maximum matchings in bipartite graphs. Let $G = (A \cup P,\mathcal{E})$ be a bipartite graph and let $M$ be a maximum matching in $G$. The matching $M$ defines a partition of the vertex set $A \cup P$ into three disjoint sets. 
A vertex $v \in A \cup P$ is even (resp. odd) if there is an even (resp. odd) length alternating path with respect to $M$ from an unmatched vertex to $v$. A vertex $v$ is unreachable if there is no alternating path from an unmatched vertex to $v$. The even, odd and unreachable vertices are denoted by $E$, $O$ and $U$ respectively. The following lemma is well known in matching theory. The proofs can be found in \cite{IrvingKMMP06}.

\begin{lemma}
Let $E$, $O$ and $U$ be the sets of vertices defined as above by a maximum matching $M$ in $G$. Then,
\begin{enumerate}
\item $E$, $O$ and $U$ are pairwise disjoint, and independent of the maximum matching $M$ in $G$.
\item In any maximum matching of $G$, every vertex in $O$ is matched with a vertex in $E$, and every vertex in $U$ is matched with another vertex in $U$. The size
of a maximum matching is $|O| + |U|/2$.
\item $G$ contains no edge between a vertex in $E$ and a vertex
in $E \cup U$.
\end{enumerate} 
\end{lemma}

%\begin{fact} \label{fact1}
%Let $G = (A \cup P, \mathcal{E})$ be a bipartite graph and $M$ be a maximum matching of $G$. We consider two vertices $a$ and $p$ of $G$ such that there is no edge between $a$ and $p$ in $G$. Let $H = G \cup \{(a,p)\}$. Then the following statements are equivalent:
%\begin{enumerate}
%\item Both $a$ and $p$ belong to $E(G)$.
%\item There exists an $M$-augmenting path in $H$.
%\item $(a,p)$ is matched in every maximum matching of $H$.
%\end{enumerate}  
%\end{fact}

\subsection{Rank-Maximal Matchings}
Next we review an  algorithm by Irving et al. \cite{IrvingKMMP06} for computing  a rank-maximal matching. Let $G = (A \cup P, \mathcal{E})$
be an instance of the rank-maximal matching problem. Every edge $e=(a,p)$ has a rank reflecting its position in the preference list of applicant $a$.  $\mathcal{E}$ is the union of disjoint sets $\mathcal{E}_i$ , i.e.,  $\mathcal{E} = \mathcal{E}_1 \cup \mathcal{E}_2 \cup \mathcal{E}_3 ... \cup \mathcal{E}_r$, where $\mathcal{E}_i$ denotes the set of edges of rank $i$ and $r$ denotes the lowest rank of an edge in $G$.

\begin{definition}
The signature of a matching $M$ is defined as an $r$-tuple $\rho(M) = (x_1,..., x_r)$ where, for each $1 \leq i \leq r$, $x_i$ is the number of
applicants who are matched to their $i$-th rank post in $M$.
\end{definition}

Let $M$ and $M'$ be two matchings of $G$, with the signatures $sig(M) = (x_1,..., x_r)$ and
$sig(M') = (y_1,..., y_r)$. We say $M \succ M'$ if there exists $k$ such that $x_i = y_i$ for each $1 \leq i < k \leq r$ and $x_k > y_k$. 

\begin{definition}
A matching $M$ of a graph $G$ is called  rank-maximal  if and only if $M$ has the best signature under the ordering $\succ$ defined above.
\end{definition}

 We give a brief description of the algorithm of Irving et al. \cite{IrvingKMMP06} for computing a rank-maximal matching, whose pseudocode (Algorithm \ref{alg1}) is given below.  Let us denote $G_i = (A \cup P, \mathcal{E}_1 \cup \mathcal{E}_2 \cup ...\cup \mathcal{E}_i)$ as a subgraph of $G$ that only contains edges of rank smaller or equal to $i$. The algorithm runs in phases. The algorithm starts with $G'_1 = G_1$ and a maximum matching $M_1$ of $G_1$.  In the first phase, the set of vertices is partitioned into $E_1$, $O_1$ and $U_1$. The edges between $O_1$ and $O_1 \cup U_1$ are deleted. Since the vertices incident to $O_1 \cup U_1$ have to be matched in $G_1$ in every rank-maximal matching, the edges of rank greater than $1$ incident to such vertices are deleted from the graph $G$.  Next we add the edges of rank $2$ and call the resulting graph $G'_2$. The graph $G'_2$ may contain some $M_1$-augmenting paths. We determine the maximum matching $M_2$ in $G'_2$  by augmenting $M_1$.
In the $i$-th phase,  the vertices are partitioned into three disjoint sets $E(G'_i)$, $O(G'_i)$ and $U(G'_i)$. We delete every edge between $O_i$ and $O_i \cup U_i$. Also, we delete every edge of rank greater than $i$ incident to  vertices in $O_i\cup U_i$. Next we add the edges of rank $(i+1)$ and call the resulting graph $G'_{i+1}$. We determine the maximum matching $M_{i+1}$ in $G'_{i+1}$ by  augmenting $M_i$. $G'$ is also called the reduced graph of $G$.

\begin{algorithm}[h]
\caption{for computing a rank-maximal matching}
\label{alg1}
\begin{algorithmic}[1]

\Procedure {RankMaximalMatching}{G}
\State $G'_1 \gets G_1$
\State Let $M_1$ be any maximum matching of $G'_1$
\For {$i = 1, 2, \ldots, r$}
	\State Partition the vertices of $G'_i$ into the sets $E(G'_i)$, $O(G'_i)$ and $U(G'_i)$ 
	\State Delete all edges in $\mathcal{E}_j$ (for $j > i$) which are incident on vertices in $O(G'_i) \cup U(G'_i)$
    \State Delete all $O(G'_i)O(G'_i)$ and $O(G'_i)U(G'_i)$ edges from $G'_i$. 
	\State Add the edges in $\mathcal{E}_{i+1}$ and denote the graph as $G'_{i+1}$. 
	\State Determine a maximum matching $M_{i+1}$ in $G'_{i+1}$ by augmenting $M_i$.
\EndFor
\State \Return $M_{r}$
\EndProcedure
\end{algorithmic}
\end{algorithm}

The following invariants are proved in  \cite{IrvingKMMP06}. 
\begin{enumerate}
\item For every $1 \leq i \leq r$, every rank-maximal matching in $G_i$ is contained in $G'_i$.
\item The matching $M_i$ is rank-maximal in $G_i$, and is a maximum matching in $G'_i$.
\item If a rank-maximal matching in $G$ has signature $(s_1,..., s_i,... s_r)$, then $M_i$
has signature $(s_1,..., s_i)$.
\item The graphs $G'_i$, ($1 \leq i \leq r$) are independent of the rank-maximal matching computed by the algorithm. 
\end{enumerate}

\begin{restatable}{lemma}{lemmatwo} \label{F}
Let $G=(A \cup P, \mathcal{E})$ and $G'=(A \cup P, \mathcal{E'})$ be two bipartite graphs with ranks on the edges. Suppose  that $\mathcal{E'} \subseteq \mathcal{E}$. Also, every edge $e \in \mathcal{E'}$ has the same rank in $G$ and $G'$. Then any rank-maximal matching $M$ of $G$ such that $M \subseteq \mathcal{E'}$ is also a rank-maximal
matching of $G'$.
  
  \end{restatable}

\begin{proof}
Let $M$ be a rank-maximal matching of $G$. Moreover, we assume that $M$ is a matching of $G'$. Since each edge of $\mathcal{E}'$ has the same rank in both $G$ and $G'$, the signature of $M$ is the same in both graphs. Therefore the signature of a rank-maximal matching of $G'$ is not worse than the signature of a rank-maximal matching of $G$. Suppose the signature of a rank-maximal matching of $G'$ is strictly better than the signature of a rank-maximal matching of $G$. If $M'$ is a rank-maximal matching of $G'$, by the construction of $G$ and $G'$, $M'$ is a matching of $G$ and $M'$ has the same signature in both $G$ and $G'$. But $M \succ M'$. Thus $M$ is not a rank-maximal matching of $G$ which is a contradiction.  
\qed
\end{proof}

\section{Properties of a preference list and strategy `best nonfirst'} \label{true}
Here we note down some properties of the preference list of any applicant.  Let us assume that the preference list of $a_1$ in $G$
has the form $
( P_1, \hspace{.2cm}P_2,\hspace{.2cm}P_3, \ldots,  P_i, \ldots, P_t)
$, where $P_i$ denotes the set of posts of rank $i$ in the preference list of $a_1$. $G \setminus \{a_1\}$ denotes the graph obtained from $G$ after the removal of the vertex $a_1$ from $G$. We define an $f$-post of $G$ in a similar way as in the popular matching problem \cite{AbrahamIKM07}

\begin{definition} \label{fpost}
A post is called an $f$-post of $G$ if and only if it  belongs to $O(G_1 \setminus \{a_1\})$ or $U(G_1 \setminus \{a_1\})$, where $G_1=(A \cup P, \mathcal{E}_1)$. The remaining posts of $G$ are called non-$f$-posts. 
\end{definition}
\begin{restatable}{lemma}{lemmathree} \label{basic}
If $P_1$ contains a post that is a non-$f$-post, then $a_1$ is always  matched to one of such posts in a rank-maximal matching of $G$ and thus to one of his first choices.
\end{restatable}

\begin{proof}
%Let us assume that the set $P_1$ contains a post that is not an $f$-post. We are going to prove that $a_1$ does not belong to $E(G_1)$. This will prove that $a_1$ is matched in every maximum matching of $G_1$, because by Lemma \ref{} point 2 a vertex belonging to $O(G) \cup E(G)$ is matched in every maximum matching of $G$.  Since every rank-maximal matching of $G$ contains some maximum matching of $G_1$, this will mean that $a_1$ is matched to a post from $P_1$ in every rank maximal matching. 

Let $M'$ be a maximum matching of $G_1 \setminus \{a_1\}$. $M'$ is a matching of $G_1$ but not necessarily of maximum size. $a_1$ is unmatched in $M'$ and $P_1$ contains a post $p$ that is not an $f$ post. Hence, $p$ belongs to $E(G_1 \setminus \{a_1\})$ and there exists an even length $M'$-alternating path $S$  starting at $p$ and ending at some unmatched vertex $p'$ in $G_1 \setminus \{a_1\}$. Therefore, $S$ together with the edge $(a_1, p)$ forms an augmenting path in the graph $G_1$.  If we apply any such augmenting path we obtain a maximum matching of $G_1$ and $a_1$ is matched in every maximum matching of $G_1$. Let us also notice that no edge $(a_1, p_1)$ such that $p_1$ is an $f$-post from $P_1$ belongs to an $M'$-augmenting path. This shows that $a_1$ is matched in every maximum matching of $G_1$ and to a non-$f$-post from $P_1$.   This completes our proof.   
\qed

\end{proof}

Next lemma shows that if $a_1$ is not matched to a rank one post in some rank-maximal matching of $G$, then an $f$-post may be defined
in an alternative way that takes into account the whole graph $G_1$. This property is needed during the construction of strategy `min max'.

\begin{restatable}{lemma}{lemmafour} \label{flemma}
Let us assume that $a_1$ is not matched to a rank one post in some rank-maximal matching of $G$. Then a post is an $f$-post if and only if it belongs to $O(G_1)$ or $U(G_1)$.
\end{restatable}

\begin{proof}
We say that a vertex $v$ has the same type in graphs $G$ and $H$ if $v \in (E(G) \cap E(H)) \cup (O(G) \cap O(H)) \cup (U(G) \cap U(H))$.

 We have assumed that $a_1$ is not matched to his  rank one post in every rank-maximal matching. By the properties of the rank-maximal matchings, $a_1 \in E(G_1)$. Hence, we can find a maximum matching $M$ of $G_1$ in which $a_1$ is an unmatched vertex.  Notice that $M$ is a maximum matching of both $G_1$ and $G_1 \setminus \{a_1\}$. Hence, a vertex that is reachable from a free vertex other than $a_1$ in $G_1$, has the same type in both $G_1$ and $G_1 \setminus \{a_1\}$. The vertices from $P$ that are reachable only from $a_1$ by an alternating path in $G_1$, belong to $O(G_1)$. These vertices become unreachable in $G_1 \setminus \{a_1\}$. Finally an unreachable vertex in $G_1$ is also an unreachable vertex in $G_1 \setminus \{a_1\}$. Therefore, a post that belongs to $O(G_1)$ or $U(G_1)$, is an $f$-post. 
 
 Conversely, let us consider an $f$-post $p$. $p$ is either an odd or an unreachable vertex in $G_1 \setminus \{a_1\}$. Suppose $p$ becomes an even vertex in $G_1$. We have proved in the previous part that $p$ must be reachable from vertices other than $a_1$. Also we know that a vertex that is reachable from a free vertex other than $a_1$ in $G_1$, has the same type in both $G_1$ and $G_1 \setminus \{a_1\}$. Hence, $p$ is an even vertex in $G_1 \setminus \{a_1\}$, which is a contradiction. Therefore, $p$ is either an odd or an unreachable vertex in $G_1$.   
 \qed

\end{proof}

 The  lemma below characterises the set of potential posts $a_1$ can be matched to, if he provides his true preference list. 

\begin{restatable}{lemma}{lemmafive}\label{flemma2}
Let $G$ be a bipartite graph and $i$ be  the rank of the highest ranked non-$f$-post in the preference list of $a_1$.  If $a_1$ is not matched to a rank one post, then $a_1$ can only be matched to a post of rank $i$ or greater  than $i$ in any rank-maximal matching of $G$.  
\end{restatable}

\begin{proof}
Let $M$ be a rank-maximal matching of $G$. While computing a rank-maximal matching of $G$ we start by finding a maximum matching of $G_1$. Since $a_1$ is not matched to a rank one post in every rank-maximal matching of $G$, by Lemma \ref{flemma}, the set of $f$-posts contains every vertex from $O(G_1)$ and $U(G_1)$.  Hence, we delete every edge, that has  rank bigger than $1$, incident to an $f$-post. Thus, every edge $e=(a_1, p)$ such that $p$ is an $f$-post and $rank(a,p) > 1$  gets deleted after the first iteration of the algorithm. Therefore, no such edge can belong to a rank-maximal matching and  $a_1$ can only be matched to a post of rank $i$ or worse.\qed

\end{proof}

The above lemmas provide us with an easy method of manipulation that guarantee that  $a_1$ can always be matched to the best non-$f$-post in his true preference list. Lemma \ref{flemma2} shows  that the most preferred non-$f$-post of $a_1$ is ranked not worse than the second most preferred post he can be matched to, when he is truthful. We assume that $a_1$ is not matched to a rank one post in every rank-maximal matching of $G$. Otherwise, the manipulator has no incentive to cheat. Let $p_i \in P_i$ be a highest ranked non-$f$ post in the true preference list of $a_1$. We put $p_i$ as a rank $1$ post in the falsified preference list of $a_1$. Next, we fill the falsified preference list of $a_1$ arbitrarily. This completes the description of strategy `best nonfirst'.

 \begin{algorithm} [!ht]
\caption{Strategy `best nonfirst'}
\label{alg2}
\begin{algorithmic}[1]
\State $p_i \leftarrow$  a highest ranked non-$f$-post in the true preference list of $a_1$.
\State $p_i \leftarrow$ the rank one post in the falsified preference list of $a_1$ in $H$
\State Fill the rest of the preference list of $a_1$ in an arbitrary order 
\State Output $H$
\end{algorithmic}
\end{algorithm}

\begin{theorem}
The graph $H$ computed by Algorithm \ref{alg2} is a strategy `best nonfirst'.
\end{theorem}
The correctness of Algorithm \ref{alg2} follows from Lemma \ref{basic}.
\section{Example of Strategy `best non-first' Not Being Optimal}\label{appendixexample}

We have given a strategy `best nonfirst' in the previous part of the paper. But that strategy may not be  optimal in the sense that the manipulator may arrange to get a post of even better rank. Let us consider an example from Figure $1$. Let us assume that $a_1$ is a manipulator. The first preference table contains  true preference list of all applicants $a_1, \ldots, a_6$. $a_1$ is matched to $p_5$ in every rank-maximal matching of this instance. $p_3$ is the best non-$f$-post in the preference list of $a_1$. The second preference table contains a falsified preference list of $a_1$.  Here $a_1$  adopted the strategy  `best nonfirst' and as a result he is matched to $p_3$ in every rank-maximal matching. We can fill the rest of the preference list of $a_1$ arbitrarily. Consider now the third preference table, in which $a_1$ falsifies his preference list in yet another way. By presenting this falsified preference list $a_1$ contrives to get matched to $p_2$ in every rank-maximal matching. Post $p_2$ is also his true first choice. Hence getting matched to $p_3$ is not an optimal strategy for $a_1$.

\begin{figure}[!ht]
\centering 
\begin{minipage}[b]{0.4\textwidth}
\begin{align*} 
&a_1 \hspace{.6cm} p_2 \hspace{.2cm}p_1 \hspace{.2cm}p_3 \hspace{.2cm}\underline{p_5} \hspace{.2cm}p_4\\
&a_2 \hspace{.6cm} \underline{p_1} \hspace{.2cm}p_2 \hspace{.2cm}p_3 \hspace{.2cm}p_4 \hspace{.2cm}p_5\\
&a_3 \hspace{.6cm} p_1 \hspace{.2 cm}p_2 \hspace{.2 cm}p_3 \hspace{.2 cm}\underline{p_4} \hspace{.2 cm}p_5\\
&a_4 \hspace{.6cm} p_1 \hspace{.2cm}p_2 \hspace{.2cm}\underline{p_3} \hspace{.2cm}p_4 \hspace{.2cm}p_5\\
&a_5 \hspace{.6cm} \underline{p_2} \hspace{.2cm}p_1 \hspace{.2cm}p_3 \hspace{.2cm}p_6 \hspace{.2cm}p_4 \hspace{.2cm}p_5\\
&a_6 \hspace{.6cm} \underline{p_6}
\end{align*}
\end{minipage}
\begin{minipage}[b]{0.4\textwidth}
\begin{align*} 
&a_1 \hspace{.6cm} \underline{p_3}\\
&a_2 \hspace{.6cm} \underline{p_1} \hspace{.2cm}p_2 \hspace{.2cm}p_3 \hspace{.2cm}p_4 \hspace{.2cm}p_5\\
&a_3 \hspace{.6cm} p_1 \hspace{.2 cm}p_2 \hspace{.2 cm}p_3 \hspace{.2 cm}\underline{p_4} \hspace{.2 cm}p_5\\
&a_4 \hspace{.6cm} p_1 \hspace{.2cm}p_2 \hspace{.2cm}p_3 \hspace{.2cm}p_4 \hspace{.2cm}\underline{p_5}\\
&a_5 \hspace{.6cm} \underline{p_2} \hspace{.2cm}p_1 \hspace{.2cm}p_3 \hspace{.2cm}p_6 \hspace{.2cm}p_4 \hspace{.2cm}p_5\\
&a_6 \hspace{.6cm} \underline{p_6}
\end{align*}
\end{minipage}
\begin{minipage}[b]{0.4\textwidth}
\begin{align*} 
&a_1 \hspace{.6cm} \underline{p_2} \hspace{.2cm}p_1 \hspace{.2cm}p_6 \hspace{.2cm}p_3 \hspace{.2cm}p_4 \hspace{.2cm}p_5\\
&a_2 \hspace{.6cm} \underline{p_1} \hspace{.2cm}p_2 \hspace{.2cm}p_3 \hspace{.2cm}p_4 \hspace{.2cm}p_5\\
&a_3 \hspace{.6cm} p_1 \hspace{.2 cm}p_2 \hspace{.2 cm}p_3 \hspace{.2 cm}p_4 \hspace{.2 cm}\underline{p_5}\\
&a_4 \hspace{.6cm} p_1 \hspace{.2cm}p_2 \hspace{.2cm}p_3 \hspace{.2cm}\underline{p_4} \hspace{.2cm}p_5\\
&a_5 \hspace{.6cm} p_2 \hspace{.2cm}p_1 \hspace{.2cm}\underline{p_3} \hspace{.2cm}p_6 \hspace{.2cm}p_4 \hspace{.2cm}p_5\\
&a_6 \hspace{.6cm} \underline{p_6}
\end{align*}
\end{minipage}
\caption{Example that shows that strategy `best non-first' may not be optimal. The underlined posts are those matched  to the corresponding applicant. We can fill the rest of the positions arbitrarily because the presence of those posts will not affect the form of any rank-maximal matching}
\end{figure}

\section{Strategy `min max'}

The example in the previous section clearly shws that the strategy `best nonfirst' may not provide an optimal solution. In this section we introduce the strategy `min max' that optimizes  the worst post $a_1$ can be matched to in a rank-maximal matching.

\subsection{Critical Rank} \label{critical}

The notion that is going to be very useful while constructing a preference list is that of a critical rank. 

\begin{definition} \label{cr}
Let $G=(A \cup P, \mathcal{E})$ be a bipartite graph with ranks on the edges belonging to $\{1,2, \ldots, r\}$. Suppose that $a \in A, p \in P$ and $(a,p)$ does not belong to $\mathcal{E}$. Let $H=(A \cup P, \mathcal{E} \cup \{(a,p)\})$. We define a critical rank of $(a,p)$ in $H$ as follows. 

If there exists a natural number   $1 \leq i \leq r$  such that  $a \in O(G'_i) \cup U(G'_i)$ or $p  \in O(G'_i) \cup U(G'_i)$, then 
the critical rank of $(a,p)$ in $H$ is equal to $\min \{i: (O(G'_i) \cup U(G'_i)) \cap \{a,p\} \neq \emptyset\}$. Otherwise, the critical rank of $(a,p)$ is defined as $r+1$.
\end{definition}

The next lemma reveals an interesting property of the critical rank of an edge $(a,p)$.

\begin{restatable}{lemma}{lemmasix} \label{crl}
Let $G,H$ and $(a,p)$ be as in Definition \ref{cr}.  Then the critical rank of $(a,p)$ is $c$ if and only if 
\begin{enumerate}
\item for every $1 \leq i <c$, the edge $(a,p)$ belongs to every rank-maximal matching of  $H_i$, in which $(a,p)$ has rank $i$, and
\item for every $c < i \leq r$, the edge $(a,p)$ does not belong to any rank-maximal matching of $H_i$, in which $(a,p)$ has rank $i$, and
\item there exists a  rank-maximal matching $M$ of $H_c$, in which $(a,p)$ has rank $c$ such that  $(a,p)$ is not contained in $M$.
\end{enumerate}
\end{restatable}

\begin{proof}
We start by proving the following claim.

\begin{claim}
Suppose that $e=(a,p)$ has rank $i$ in the graph $H$. Then $(a,p)$ belongs to every rank-maximal matching of $H_i$ if and only if for every
$j \leq i$ both $a \in E(G'_j)$ and $p \in E(G'_j)$.
\end{claim}

Since $(a,p)$ has rank $i$ in $H$ and $G$ differs from $H$ only by the existence of the edge $(a,p)$, we know that the reduced graphs $G'_j$
and $H'_j$ are the same for each $j<i$. Also, the reduced graphs $H'_i$ and $G'_i$ may be the same or differ by the existence of the edge $e$. 

If for some $j<i$ it holds that $a \in O(G'_j) \cup U(G'_j)$ or $p \in O(G'_j) \cup U(G'_j)$, then the edge $(a,p)$ does not belong to $H'_i$, because it is removed at the beginning of phase $j+1$ during the computation of a rank-maximal matching of $H$. Thus in this
case $e$ does not belong to any rank-maximal matching of $H_i$. If for every $j <i$, $a \in E(G'_j)$ and $p \in E(G'_j)$ and  $a \in  O(G'_i) \cup U(G'_i)$  or $p \in  O(G'_i) \cup U(G'_i)$, then 
$e$ belongs to $H'_i$, but there exists a rank-maximal matching of $H_i$ that does not contain $e$. This is so because the addition of $e$ to $G'_i$ does not create any augmenting path in $H'_i$, therefore a maximum matching  of $G'_i$ is also a maximum matching of $H'_i$ and thus
a rank-maximal matching of $H_i$. This shows that in this case every rank-maximal matching of $G_i$ is also rank-maximal in $H_i$. Hence, there exists a rank-maximal matching of $H_i$ that does not contain $(a,p)$.

Assume now that for every $j \leq i$ both $a \in E(G'_j)$ and $p \in E(G'_j)$.  This means that the reduced graph $H'_i$ does contain $(a,p)$.
Any rank-maximal matching of $H_i$ is a maximum matching of $H'_i$. Let $M_i$ denote a maximum matching of $G'_i$. Since $a \in E(G'_i)$
and $p \in E(G'_i)$, $(a,p)$ belongs to every $M_i$-augmenting path - because $a$ is the endpoint of some even length $M_i$-alternating path
ending at a free vertex $a' \in A$ and similarly, $p$  is the endpoint of some even length $M_i$-alternating path
ending at a free vertex $p' \in P$. Together with $(a,p)$ these paths form an $M_i$-augmenting path in $H'_i$. Additionally, we notice that $(a,p)$ belongs to every $M_i$-augmenting path and thus to every maximum matching of $H'_i$ and hence to every rank-maximal matching of $H_i$.

Let us now prove the other direction of the claim and suppose that $(a,p)$ belongs to every rank-maximal matching of $H_i$. This means that $(a,p)$ belongs to every maximum matching of $H'_i$. Then $(a,p)$ must be contained in $H'_i$ and by the above arguments, we know that for every $j < i$ both $a \in E(G'_j)$ and $p \in E(G'_j)$. No maximum matching of $G'_i$ contains $(a,p)$ - therefore a maximum matching of $H'_i$ must be bigger by one than a maximum matching of $G'_i$. This means that $(a,p)$ must belong to a path augmenting with respect to a maximum matching of
$G'_i$, which means that the endpoints of $(a,p)$ belong to $E(G'_i)$. This way we have proved the claim.

The claim, which we have just proved, shows that if (i) for every $1 \leq i <c$ the edge $(a,p)$ belongs to every rank-maximal matching of  $H_i$, in which $(a,p)$ has rank $i$ and (ii) there exists a  rank-maximal matching $M$ of $H_c$, in which $(a,p)$ has rank $c$ such that  $(a,p)$ is not contained in $M$, then $c= \min \{i: (O(G'_i) \cup U(G'_i)) \cap \{a,p\} \neq \emptyset\}$ and hence, $c$ is the critical rank of $(a,p)$.

The claim also shows that if $c$ is the critical rank of $(a,p)$, then for every $i <c$ the edge $(a,p)$ belongs to every rank-maximal matching of $H_i$, in which $(a,p)$ has rank $i$ and there exists a rank-maximal matching of $H_c$, in which $(a,p)$ has rank $c$ that does not contain $(a,p)$. 
It remains to prove that if $c$ is the critical rank of $(a,p)$, then for every $c < i \leq r$ the edge $(a,p)$ belongs to no rank-maximal matching of $H_i$, in which $(a,p)$ has rank $i$. This follows from the fact that any edge of rank $i>c$ incident to a vertex belonging to $O(H'_c) \cup U(H'_c)$ is removed from $H'_i$ and therefore cannot belong to a maximum matching of $H'_i$ and thus cannot be present in any rank-maximal matching of $H_i$. This ends the proof of the lemma.\qed

\end{proof}

\begin{corollary} \label{corollary1}
The critical rank of an edge incident to an $f$-post of $G$ is $1$ in $G$.
\end{corollary}

 The next two lemmas explain the change of the critical rank of $(a,p)$ when we add an $f$-post $p'$ as a rank $1$ post to the preference list of $a$.

\begin{restatable}{lemma}{lemmaseven} \label{decrease}
Let $G$ be a bipartite graph and $a$ be an applicant. Let $p$ be the only post in the preference list of $a$. Suppose that the critical rank of $(a,p)$ is $c$ in $G$. Let  $\hat{G} = G \cup \{(a,p')\}$ where $p'$ is a rank $1$, $f$-post in the preference list of $a$. Then the critical rank of $(a,p)$  is at most $c$ in $\hat{G}$.
\end{restatable}

\begin{proof}

It suffices to show that $(a,p)$ is not matched in every rank-maximal matching of $\hat{G}_c$ with $(a,p)$ having rank $c$ and for every $c <i \leq r$ no rank-maximal matching of $\hat{G}$, in which $(a,p)$ has rank $i$  contains $(a,p)$. 

Suppose, $(a,p)$ is matched in every rank-maximal matching of $\hat{G}_c$. If $M$ is a matching of $\hat{G}_c$, by Lemma \ref{F}, $M$ is also a rank-maximal matching of $G$. Since the critical rank of $(a,p)$ is $c$ in $G$, there exists a rank-maximal matching $M'$ of $G_c$ that does not contain the edge $(a,p)$. The signature of $M$ and $M'$ is the same. Hence $M'$ is a rank-maximal matching of $\hat{G}_c$, which is a contradiction.

To prove the second part, suppose to the contrary that there exists rank-maximal matching $M$ of $\hat{G}$, in which $(a,p)$ has rank $i$ that  contains $(a,p)$.
Then $M$ is also a matching of $G$ and by Lemma \ref{F} it is also a rank-maximal matching of $G$. However, by Lemma \ref{crl} no rank-maximal matching of $G$,  in which $(a,p)$ has rank $i>c$  can contain $(a,p)$ - a contradiction. \qed

\end{proof}

\begin{restatable}{lemma}{lemmaeight}
Let $G=(A \cup P,\mathcal{E})$ be a bipartite graph, in which $a$ has two neighbors $p$ and $p'$ such that apart from $(a,p)$, each edge has a rank belonging to $\{1,2, \ldots, r\}$. Additionally,  $p'$ is a  rank one $f$-post in the preference list of $a$. Let $G'=(A \cup P, \mathcal{E} \setminus \{(a,p)\})$ and $G''= (A \cup P, \mathcal{E} \setminus \{(a,p')\})$. Suppose that $a$ becomes  unreachable in $G'_i$ and  the critical rank of $(a,p)$ is $c$ in $G''$. Then the critical rank of $(a,p)$ in the graph $G$ is equal to, correspondingly:
\begin{enumerate}
\item  $c$  if $c \leq i$,
\item  $i$  if $c >i$.
\end{enumerate}
\end{restatable}

\begin{proof}

We can prove that for every $j<i$ there exists a rank-maximal matching of $G'_j$ that contains $(a,p')$ and there exists a rank-maximal matching of $G'_j$ that does not contain
$(a,p')$.

\begin{claim} \label{C1}
Let us suppose that $(a,p)$ has rank $c'<i$ in $G$. Then, the existence of a rank-maximal matching $M$ of $G$ that does not contain $(a,p)$ implies that 
the critical rank of $(a,p)$ in $G''$ is at most $c'$.
\end{claim}

\begin{proof}
 By Lemma \ref{F}, $M$ is a rank-maximal matching of $G'$ and hence, every rank-maximal matching of $G'$
is also rank-maximal in $G$. We know that there exists a rank-maximal matching $M'$ of $G'$ that does not contain $(a,p')$. Thus $a$ is unmatched in $M'$. By Fact
\ref{F}, $M'$ is also a rank-maximal matching of $G''$. $M'$ does not contain $(a,p)$, which shows that the critical rank of $(a,p)$ is at most $c'$ in $G''$. \qed

\end{proof}

First we assume that $c \leq i$. Since $G''$ is a subgraph of $G$, by Lemma \ref{decrease}, the critical rank of $(a,p) \leq c$ in $G$. Suppose the critical rank of $(a,p)=c'<c$ in $G$.  Let us consider a graph $G$, in which $(a,p)$ has rank $c'$. Since the critical rank of $(a,p)$ is equal to $c'$, there exists a rank-maximal
matching $M$ of $G$  that does not contain $(a,p)$. By the above claim,  the critical rank of $(a,p)$ in $G''$ is at most $c' <c$ - a contradiction. We  conclude that the critical rank of $(a,p)$ remains $c$ in $G$ if $c \leq i$.

Let us consider  now the case when $c > i$. First we show that the critical rank of $(a,p) \leq i$ in $G$. Let us consider the graph $G'$. The vertex $a$ becomes unreachable in $G'$ after iteration $i$. We know that $G = G' \cup \{(a,p)\}$. Hence, the edge $(a,p)$ is deleted in the graph $G$ if the rank of $(a,p) > i$ in $G$. Therefore, the critical rank of $(a,p) \leq i$ in $G$. 

Next we show that the critical rank of $(a,p) = i$ in $G$. Suppose the critical rank of $(a,p)$ equals $i' < i$ in $G$. Since the critical rank of $(a,p)$ is equal to $i'<i$, there exists a rank-maximal
matching $M$ of $G$, in which $(a,p)$ has rank $i'$  that does not contain $(a,p)$. Again by the claim, the critical rank of $(a,p)$ in $G''$ is at most $i' <c$ - a contradiction. Therefore we have proved that the critical rank of $(a,p)$ is equal to $i$ in $G$. \qed

\end{proof}

\begin{corollary} \label{trunc}
Let $G$ be a bipartite graph such that $p'$ is a  rank one, $f$-post in the preference list of $a$. Suppose  that $a$ becomes  unreachable  after iteration $i$. Let  $\hat{G} = G \cup \{(a,p)\}$ with the rank of the edge $(a,p)$ being $c$. Then $(a,p)$ is never matched in a rank- maximal matching of $\hat{G}$, if $c > i$.
\end{corollary}

The next lemma is useful while building a falsified preference list of $a$ using the strategy `min max'. This lemma basically combines  two short preference lists of $a$ into a longer preference list.

\begin{restatable}{lemma}{lemmanine} \label{comb}
Let us consider two bipartite graphs $G_1$ and $G_2$ such that  $p$ is a rank one, $f$-post  in the preference list of $a$ in  both graphs. Also $a$ has only two neighbors in each of the graphs. In  $G_1$, $a$ has $p_1$ as a rank $i$ post. In $G_2$, $a$ has $p_2$ as a  rank $j$ post.   We assume that $G_3$ is the union of graphs $G_1$ and $G_2$. Then $a$ is matched to $p$ in every rank-maximal matching of both $G_1$ and $G_2$ if and only if $a$ is matched to $p$ in every rank-maximal matching of $G_3$.
\end{restatable}

\begin{proof}
   Assume  that $a$ is matched to $p_1$ in a rank-maximal matching $M$ of $G_3$.   $M$ is also a matching in the graph $G_1$. Since $G_1$ is a subgraph of $G_3$, from the fact \ref{F}, $M$ is a rank-maximal matching of $G_1$. This means that $a$ is matched to $p_1$ in some rank-maximal matchings of $G_1$, which is a contradiction. 

Conversely, let $a$ be matched to $p$ in every rank-maximal matching of $G_3$. Let $M_3$ be a rank-maximal matching of $G_3$. Without loss of generality, suppose that $a$ is matched to $p_1$ in a rank-maximal matching $M_1$ of $G_1$.  Since $G_1$ is a subgraph of $G_3$, from fact \ref{F}, $M_3$ is a rank-maximal matching of $G_1$. Thus, $M_1$ and $M_3$ have the same signature. Therefore, $M_1$ is a rank-maximal matching of $G_3$, which is a contradiction.  
\qed
\end{proof}

\subsection{Algorithm for Strategy `min max'}\label{minmax}
In this section, we give an algorithm that computes a graph $H$ by using the strategy `min max' for the applicant $a_1$. We recall that strategy `min max' consists in finding a full preference list for $a_1$ such that the maximal rank
of a post he can obtain is minimized. Since we have assumed that $a_1$ is not always matched to his first choice  when he is truthful
and since strategy `best nonfirst' ensures that $a_1$ always gets the highest ranked non-$f$-post, it remains to check if it is possible
for $a_1$ to get one of the $f$-posts in every rank-maximal matching. For a given $f$-post $p$ we want to verify if $a_1$ can  construct  a full preference list   that  guarantees that $a$ becomes matched to $p$  in every rank-maximal matching of the resultant graph. From all such $f$-posts, we want to choose that of the highest rank in the true preference list of $a_1$. Below we show that this way we indeed compute the strategy `min max'.

% We want to choose the best ranked $f$ post in the true preference list of $a_1$ that satisfies the above property. We will prove later that this is an optimal improvement. We denote the graph with a falsified preference list of $a_1$ be $H$. Also, $a_1$ wants to write his unwanted posts in the falsified preference list as late as possible. So it is beneficial for $a_1$ to make a strict falsified preference list.  

Let $p$ be an $f$-post that $a_1$ wants to be matched to in every rank-maximal matching of $H_p$, where $H_p$ contains a full falsified preference list of $a_1$. How do we construct such $H_p$?  Let  $\hat{H}_p$ denote the graph, in which $a_1$ is incident only  to $p$ and $(a_1,p)$ has rank one. 
By Lemma \ref{comb}, we know that in order to obtain $H_p$, it suffices to find a certain number of graphs $H_{p,p_j}$ such that $p$ and $p_j$
are the only posts in the preference list of $a_1$, $p$ has rank $1$, $p_j$ has rank $j>1$ and every rank-maximal matching of $H_{p,p_j}$ matches $a_1$ to $p$. Then we can combine those graphs into one graph $H_p$. In fact, it suffices to fill the preference list of $a_1$ only till rank $k$, where $k$ is the rank, when $a_1$ becomes an unreachable vertex in $\hat{H'}_{p,i}$.  This follows from Corollary \ref{trunc}, which says that no rank-maximal matching of $H_{p,p'}$ such that $(a_1,p')$ has rank $i>k$ contains $(a_1, p')$.
Therefore, the ranks greater than $k$ in the preference list of $a_1$ may be filled with arbitrary posts not occurring previously.

Suppose that  we want to find a ``good" post for rank $i < k $ in the preference list of $a_1$. First, we check if there is any available post $p'$ such that the critical rank of $(a_1,p')$ is smaller than $i$ in ${H}_{p, p'}$. If we find such a post, then by Lemma \ref{crl}, the edge $(a_1,p')$
never occurs in a rank-maximal matching of $H_{p,p'}$ in which $(a_1,p')$ has rank $i$. Therefore, we may add $p'$ to the preference list of $a_1$ as a rank $i$ post.  Otherwise, we consider a post $p''$ with critical rank $i$ in the graph $H_{p, p''}$. We   verify if 
$p$ is matched to $a_1$ in every rank-maximal matching, when we add $(a_1, p'')$ as a rank $i$ edge to the graph $\hat{H}_p$. If yes, then we put $p''$ as an $i$th choice in $H_p$. If not, we check another  post with  critical rank $i$. If we are unable to find  any  post for rank $i$, Algorithm \ref{HP} outputs that there does not exist any preference list that matches $a_1$  to $p$ in every rank-maximal matching of $H_p$. 

The algorithm that computes a graph $H_p$, if it exists, is given below as Algorithm \ref{HP}. The thing that still requires explanation is how we verify if $(a,p)$ belongs to every rank-maximal matching of $H_{p,p'}$. For this, we need the reduced graph of $H_{p,p'}$ from phase $r$, which we can obtain by either applying the
standard rank-maximal matching algorithm \cite{IrvingKMMP06} or we can use one of the dynamic algorithms\cite{GhosalKP17}\cite{NimbhorkarV17} if we want to have a faster algorithm. Once we have access to this
reduced graph of $H_{p,p'}$ we can use  the following lemma.

\begin{restatable}{lemma}{lemmaten}\label{switch}
Let $G$ be an instance of the rank-maximal matching problem, in which the maximal rank of 
an edge is $r$. Also, we assume that $M$ is a fixed rank-maximal matching of $G$ that 
matches an edge $(a,p)$. Let us consider the switching graph of the matching $M$ in $G$. 
Then the edge $(a,p)$ belongs to every rank-maximal matching of $G$ if there does not exist 
any switching path or switching cycle in the switching graph of $M$ that contains the 
vertex $p$.
\end{restatable}

\begin{proof}
Let us fix a rank-maximal matching $M$ of $G$ that matches the edge $(a,p)$. Theorem $1$ from \cite{ghosal2014rank} states that every rank-maximal matching $G$ can be obtained from $M$ by applying some vertex-disjoint switching paths and switching cycles in the switching graph of $M$. If there does not exist any switching path or switching cycle containing the vertex $p$, $p$ has the same partner in every rank-maximal matching of $G$. Therefore, $(a,p)$ is matched in every rank-maximal matching of $G$. \qed

\end{proof}

%Otherwise, we continue for the next iteration. If the algorithm fails for every $f$ post in the preference list of $a_1$, we say that the best ranked non-$f$ post is the optimal solution. (See the easy manipulation in section \ref{}). The algorithm uses the dynamic rank maximal matching algorithm $O(n^2)$ times. The running time of the dynamic rank maximal matching is $O(m)$. So the total running time of this cheating strategy is $O(n^2m)$. We give a pseudocode of the algorithm in \ref{}.

\begin{definition}
We say that an $f$-post $p$ is {\em feasible} if there exists a graph $H_p$ such that every rank-maximal matching of $H_p$ matches $a_1$ to $p$.
\end{definition}

\begin{algorithm} 
\caption{Construction of $H_p$}
\label{HP} 
\begin{algorithmic}[1]
%\State $C \leftarrow \emptyset$ \  (the set of non-alive vertices)
\State $C_i \leftarrow \{p' \in P:$ the critical rank of $(a_1, p')$ in $H_{p,p'}$ equals $i \}$
\State $L \leftarrow$ an empty list  - $L$ is the falsified preference list of $a_1$ that is going to have the form $(p_1, p_2,...,p_n)$, where $p_i$ denotes the rank $i$ post in $L$.
\State add $p$ to $L$ -- this is the rank $1$ post in the preference list $L$ of $a_1$
\State $k \leftarrow$ the number of phase  when $a_1$ becomes unreachable in $\hat{H}_p$, i.e., $a_1 \in U(\hat{H'}_{p,k})$ and $a_1 \in E(\hat{H'}_{p,i})$ for every $i<k$, where $\hat{H'}_{p,i}$ is the $i$-th reduced graph of $\hat{H}_p$.
%\State $H_{p'} \leftarrow$ is the bipartite graph such that $a_1$ has two posts $p$ and $p'$ in his preference list
\State $C \leftarrow  C_1$ 
\For{$i = 2, \ldots, k$}
	\If {$C \neq \emptyset$} \ (there exists a post $p'$ in $C$)
    	\State add $p'$ as a rank $i$ post to the falsified preference list $L$ of $a_1$
			 \State $C \leftarrow C \setminus \{p'\}$
    \Else \ \ ($C = \emptyset$)
		\State $SEARCH \leftarrow TRUE$
		\While {$\exists p'$ with critical rank of $(a,p')$  equal to $i$ in $\hat{H}_p$ \ and \   $SEARCH$}    
        
        \If {$(a,p)$ belongs to every rank-maximal matching of $H_{p,p'}$ \ (Lemma \ref{switch})\ } 
        	\State add $p'$ as a rank $i$ post to the falsified preference list $L$ of $a_1$
					\State $SEARCH \leftarrow FALSE$
         \EndIf   
     \EndWhile
    \EndIf
		\If {SEARCH}  Break
		\EndIf
    \State $C \leftarrow C \cup C_i$
\EndFor
\If {$L$ is a full preference list} \\
\Return{$H_p$}
\Else \\
\Return{p is not a feasible $f$-post}
\EndIf
\end{algorithmic}
\end{algorithm}

In the lemma below we prove the correctness of Algorithm \ref{HP}.

\begin{restatable}{lemma}{lemmaeleven}
If Algorithm \ref{HP} outputs a graph $H_p$, then every rank-maximal matching of $H_p$ matches $a_1$ to $p$.
Otherwise, there does not exist a graph $H_p$, in which $a_1$ gives a full preference list  such that every rank-maximal matching of $H_p$ matches $a_1$ to $p$.
\end{restatable}

\begin{proof}
If Algorithm \ref{HP} outputs a graph $H_p$, then by Corollary \ref{trunc} and  Lemma \ref{comb} every rank-maximal matching of $H_p$ matches $a_1$ to $p$. If the algorithm does not output any graph, then it means that there was a problem for some $i$ with finding a rank $i$ post for the preference  list of $a_1$. However, the algorithm considers every available post $p'$ of critical rank at most $i$ in $H_{p,p'}$ for that position. Therefore, if none of them has the required  property that every rank-maximal matching of $H_{p,p'}$ matches $a_1$ to $p$, then
no post of critical rank greater than $i$ in $H_{p,p'}$ satisfies it either.  This finishes the proof of correctness of Algorithm \ref{HP}. \qed

\end{proof}

\begin{restatable}{theorem}{theoremtwo}
 Let $p$ be the highest ranked feasible $f$-post  in the true preference list of $a_1$. Then $H_p$ output by Algorithm \ref{HP} is a strategy `min max'. Moreover, each
graph $H$ that is a strategy `min max' has the property that each rank-maximal matching of $H$ matches $a_1$ to $p$.

\end{restatable}

\begin{proof}
Let $H_{opt}$ denote a graph that is a strategy `min max'.   
Suppose that $H_p$ is not a strategy `min max'. There exists then a post $p'$ such that $rank(a_1,p') < rank(a_1,p)$ in the true preference list of $a_1$ and $(a_1,p')$ belongs to some rank-maximal matching of $H_{opt}$.
Since $rank(a_1,p') < rank(a_1,p)$ in the true preference list of $a_1$, $p'$ is an $f$-post. 
Also, $a_1$ can only be matched to a post of rank   not worse than $rank(a_1,p)$ in the true preference list of $a_1$ in a rank-maximal matching of $H_{opt}$, otherwise $H_{opt}$ would not be a strategy `min max' because it would fare worse than the strategy $H_p$ in terms of the worst post $a_1$ can become matched to. 
This means that $a_1$ can only be matched to $f$-posts under strategy $H_{opt}$, because every non-$f$-post has a worse rank than the rank of post $p$ in the preference list of $a_1$. Corollary \ref{corollary1} shows that the critical rank of an $f$-post is $1$. Hence, $a_1$ can only be matched to a post that has rank $1$ in $H_{opt}$.
Therefore, we can only put $p'$ as a rank $1$ post in the preference list of $a_1$ in $H_{opt}$. 

Now we will show how to build a  preference list of $a_1$ without any ties that matches $a_1$ to $p'$ in every rank-maximal matching. Let us denote this graph as $H_{mod}$. We put $p'$ as a rank $1$ post.  We create the preference list of $a_1$ in $H_{mod}$ using the preference list of $a_1$ in $H_{opt}$. Suppose we want to add  a rank $i$ post to the preference list of $a_1$ in $H_{mod}$. If there is only one rank $i$ post in the preference list of $a_1$ in $H_{opt}$, then that post will have the same rank in the preference list of $a_1$ in $H_{mod}$. If there is more than one rank $i$ posts in the preference list of $a_1$ in $H_{opt}$, we choose only one of them. Finally, we add the remaining  posts to $a_1$'s preference list in $H_{mod}$ in an arbitrary order as his least preferred posts. This way we get a preference list
without any ties.

First, we show that $H_{mod}$ and $H_{opt}$ have the same signature. A rank-maximal matching in $H_{opt}$ that matches $a$ to $p'$ is a matching in $H_{mod}$ and has the same signature. Thus a rank-maximal matching of $H_{opt}$ does not have a better signature than a rank- maximal matching of $H_{mod}$. Also, any post in the preference list of $a_1$ has  rank in $H_{opt}$ that is not worse than in $H_{mod}$. If we consider a rank-maximal matching of $H_{mod}$, then the signature of that matching is not worse in $H_{opt}$. Therefore rank-maximal matchings of $H_{opt}$ and $H_{mod}$ have   the same signature. 

 Next we want to prove that $a_1$ is matched to $p'$ in every rank-maximal matching of $H_{mod}$. Suppose that $a_1$ is matched to some $p''$ in a rank-maximal matching $M$ of $H_{mod}$.  We can see that $M$ is a matching in $H_{opt}$. The rank of $p''$ in the preference list of $a_1$ in $H_{opt}$ is not worse than the rank of $p''$ in the preference list of $a_1$ in $H_{mod}$. This shows that the signature of the matching $M$ in $H_{opt}$ is not worse than the signature of $M$ in $H_{mod}$. But we know that the signatures of rank-maximal matching of  $H_{opt}$ and $H_{mod}$ are the same. Hence, $M$ is a rank-maximal matching of $H_{opt}$. This implies that $a_1$ is matched to $p''$ in a rank-maximal matching of $H_{opt}$ too. 

The preference list of $a_1$ in  $H_{mod}$ does not contain any ties, which means that $p''$ is not a rank $1$ post in the preference list of $a_1$ in $H_{mod}$. We know that the critical rank of an $f$-post is $1$. Therefore $p''$ is not an $f$-post. But $a_1$ cannot be matched to a non-$f$-post in $H_{opt}$ - a contradiction. This proves that $a_1$ is matched to $p'$ in every rank-maximal matching of $H_{mod}$.  Lemma \ref{comb} shows that  Algorithm \ref{HP} could also create the same preference list because it would consider each graph $H_{p,p_i}$ such that $p_i$ has rank $i$ in $H_{mod}$.  Since  the rank of $(a_1,p')$ is strictly better than the rank of $(a_1,p)$ in the true preference list of $a_1$, it means that $p$ is not a highest rank feasible $f$-post - a contradiction. Therefore $H_p$ is indeed  a strategy `min max'. Also, by the construction of $H_{mod}$ we have shown that each
graph $H$ that is a strategy 'min max' has the property that each rank-maximal matching of $H$ matches $a_1$ to the same post. \qed

 \end{proof}

% \section{Open Problems}
%One of the main open problems is to formally write a faster algorithm of the cheating strategy in rank maximal matching with the help of the lemma \ref{fastest}. The improvement in the running time will also improve the running time of the algorithm of the cheating strategy in popular matching settings. The critical rank can be used independently to solve some problems of rank maximal matchings. Using the property of the critical rank, we can check the relation between the change of the rank of an edge and the probability of the edge getting matched in a rank maximal matching.
 
 \section{`Improve Best' Strategy} \label{improvebest}
In the previous section, we present an optimal algorithm for the manipulation strategy `min max'. Also, we have shown that `min max' strategy matches $a_1$ to the same post in every rank-maximal matching. There is a possibility that $a_1$ may not be matched to his most preferred post in any rank-maximal matching by using `min max' strategy. Therefore, a natural question is whether there is a manipulation strategy that matches $a_1$ to his most preferred post in at least one rank-maximal matching. Here we present a strategy that matches $a_1$ to his most preferred post in some rank-maximal matchings.

\begin{restatable}{lemma}{lemmatwelve} \label{first}
Let $G$ be a bipartite graph and $a$, $a'$ be two applicants with identical preference list. If $(a,p)$ is matched in a rank-maximal matching of $G$ then the edge $(a',p)$ is present in the reduced graph of $G$.
\end{restatable}

 \begin{proof}
 Let ($a,p)$ is matched in $G$. If we swap the partners of $a$ and $a'$, we get another matching in $G$. Since the preference list of $a$ and $a'$ are identical,  the signature of both matchings is the same. Hence, there exists a rank-maximal matching that matches the edge $(a', p)$. Thus $(a', p)$ is present in the reduced graph of $G$. \qed
\end{proof}

\subsection{Strategy}
Here we give a brief description of the strategy `improve best' for  $a_1$. We assume that $p_1$ is the most preferred post in the preference list of $a_1$. Given a bipartite graph containing the true preference list of every applicant and a rank-maximal matching, we apply the decremental dynamic rank-maximal matching algorithm\cite{GhosalKP17} to delete $a_1$ from the graph. Let us denote the updated graph as $G$. First, we check which applicant is matched to $p_1$ in $G$. Suppose $p_1$ is matched to $a'_1$ in $G$. $a_1$ copies the preference list of $a'_1$ and present it as his falsified preference list. Let $H$ be the graph with the falsified preference list of $a_1$. We claim that $(a_1, p_1 )$ is a rank-maximal pair in $H$ when he uses this strategy. In other words, $a_1$ is matched to $p_1$ in some rank-maximal matchings of $H$.

\subsection{Correctness}
We assume that there is a rank-maximal matching of $G$ that matches the edge $(a'_1, p_1)$. We apply the incremental rank-maximal matching algorithm to get a rank-maximal matching of $H$. If we consider the vertex $a'_1$ in $H$, there are two possibilities. In the first case, let $a'_1$ be matched to  $p_1$ in $H$. Without loss of generality, we assume that $a_1$ is matched to $p'_1$ in $H$. Now consider another matching of $H$ by swapping the partners of $a_1$ and $a'_1$. Since $a_1$ and $a'_1$ has the identical preference list, both matchings have the same signature. Hence $(a_1, p_1)$ is a rank-maximal pair of $H$.
\begin{figure}[!ht]
  \centering
  
    \includegraphics[width=\textwidth]{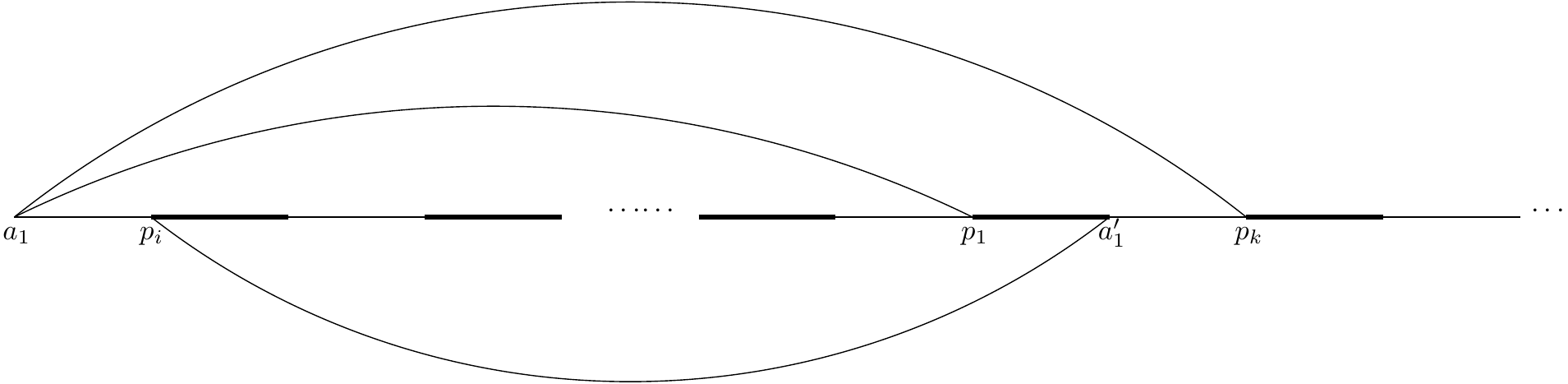}
  \caption{$a_1$ is the manipulator who is copying the preference list of $a'_1$ and the thick edges belong to the matching of $G$}

\label{bestrank}
\end{figure}
In the other case, we assume that $(a'_1, p_1)$ is not a matched edge in $H$. By Theorem 10 of \cite{GhosalKP17}, we get a rank-maximal matching of $H$ from $G$ by applying an alternating path starting from $a_1$ (Firgure \ref{bestrank}). Since $(a'_1, p_1)$
is not matched in $H$, the alternating path contains the edge $(a'_1, p_1)$. Suppose, a rank $i$ post $p_i$ (resp. rank $k$ post $p_k$) is matched to $a_1$ (resp. $a'_1$) in $H$. By Lemma \ref{first}, the edges $(a_1, p_k)$ and $(a'_1, p_i)$ are also present in the reduced graph of $H$. Therefore, the path segment $a_1 \rightarrow p_i \rightarrow \cdots \rightarrow p_1 \rightarrow a'_1 \rightarrow p_k$ together with the edge $(a_1, p_k)$ creates an alternating cycle in the reduced graph of $H$. Any alternating cycle in a reduced graph is a switching cycle \cite{ghosal2014rank}. If we switch along the alternating cycle in $H$, $(a'_1, p_1)$ becomes a  matched edge in $H$. Now we have arrived at the first case. Therefore we have proved that $(a_1, p_1)$ is a rank-maximal pair in $H$.

\bibliography{rmm_cheating}
\bibliographystyle{splncs03}

\end{document}